\documentclass[submission,copyright,creativecommons]{eptcs}

\providecommand{\event}{MEMICS 2016} 
\usepackage{breakurl}             
\usepackage{underscore}           

\title{Anti-Path Cover on Sparse Graph Classes\thanks{The work was supported by the project SVV-2016-260332.}}
\author{Pavel Dvo\v{r}\'{a}k
\institute{Computer Science Institute of Charles University\\Charles University\\ Prague, Czech Republic}
\email{koblich@iuuk.mff.cuni.cz}
\and
Du\v{s}an Knop\thanks{
Author was supported by the project GAUK 1784214 and by the project CE-ITI P202/12/G061}
\institute{Department of Applied Mathematics\\Charles University\\ Prague, Czech Republic}
\email{knop@kam.mff.cuni.cz}
\and
Tom\'{a}\v{s} Masa\v{r}\'{i}k\thanks{
Author was supported by the project GAUK 338216 and by the project CE-ITI P202/12/G061. }
\institute{Department of Applied Mathematics\\Charles University\\ Prague, Czech Republic}
\email{masarik@kam.mff.cuni.cz}
}
\def\titlerunning{Anti-Path Cover on Sparse Graph Classes}
\def\authorrunning{P. Dvo\v{r}\'{a}k, D. Knop \& T. Masa\v{r}\'{i}k}
\usepackage[utf8]{inputenc}
\usepackage[english]{babel}
\usepackage{etoolbox}
\usepackage{xspace}
\usepackage{amsthm}

\newcounter{cptTh}
\newtheorem{theorem}[cptTh]{Theorem}
\newtheorem{proposition}[cptTh]{Proposition}
\newtheorem{lemma}[cptTh]{Lemma}

\newtheorem{definition}[cptTh]{Definition}

\newcommand{\nd}{\ensuremath{\mathop{\mathrm{nd}}}}

\newcommand{\tw}{\ensuremath{\mathop{\mathrm{tw}}}}
\newcommand{\td}{\ensuremath{\mathop{\mathrm{td}}}}

\newcommand{\FPT}{\ensuremath{\mathsf{FPT}}\xspace}

\newcommand{\MSOt}{{$\mathsf{MSO}_2$}\xspace}

\newcommand{\dist}{\ensuremath{\mathop{\mathrm{d}}}}

\newcommand{\Trule}{\rule{0pt}{3ex}}
\newcommand{\Brule}{\rule[-1.5ex]{0pt}{0pt}}

\newcommand{\prob}[3]{
\begin{center}
\begin{tabular} {|ll|}
	\hline
	\multicolumn{2}{|c|}{#1\Trule} \\
	{\bf Input:\enspace}&{\parbox[t]{34em}{#2}}\\
	{\bf Question:\enspace}&\parbox[t]{34em}{#3\Brule}\\
	\hline
\end{tabular}
\end{center}
\vspace{-3pt}
}

\newcommand{\kPathCover}{{\sc $k$-Path Cover}\xspace}
\newcommand{\kAntiPathCover}{{\sc $k$-Anti-Path Cover}\xspace}
\newcommand{\cG}{\ensuremath{\overline{G}}}

\begin{document}
\maketitle

\begin{abstract}
\begin{abstract}
We show that it is possible to use Bondy-Chvátal closure to design an \FPT algorithm that decides whether or not it is possible to cover vertices of an input graph by at most $k$ vertex disjoint paths in the complement of the input graph. \\

More precisely, we show that if a graph has tree-width at most $w$ and its complement is closed under Bondy-Chvátal closure, then it is possible to bound neighborhood diversity of the complement by a function of $w$ only. \\

A simpler proof where tree-depth is used instead of tree-width is also presented.
\end{abstract}

\end{abstract}



\section{Introduction}

Graph Hamiltonian properties are studied especially in connection with graph connectivity properties. A graph is called {\em Hamiltonian} if there is a path passing through all its vertices in that graph. In this work we are interested in sparse graph setting for which this question was already solved e.g.\ using the famous theorem of Courcelle~\cite{courcelle90}. It is possible to express the Hamiltonian property by an \MSOt formula and thus resolve the question by an \FPT algorithm. For a graph~$G$ and a positive integer $k$ we say that $G$ is {\em $k$-path coverable} if there exists a collection of at most $k$ vertex disjoint paths in~$G$ such that the vertices of~$G$ are the union of vertices of all paths in the collection (that is, each vertex belongs to exactly one path in the collection). We give a simple, but interesting, twist to the question to rise a new problem:

\prob{\kAntiPathCover}{
A graph $G$ and positive integer $k$.
}{
Is the complement of graph $G$ $k$-path coverable? 
}

In this notation {\sc Hamiltonian Anti-Path} problem is exactly the {\sc $1$-Anti-Path Cover} problem.
It is not hard to see that we may focus on solving the {\sc Hamiltonian Anti-Path} as the \kAntiPathCover problem is reducible to the {\sc Hamiltonian Anti-Path} by addition of $k$ isolated vertices (apex vertices in the complement graph). Note that this does not affect tree-width nor tree-depth.
Indeed, for unrestricted graphs this question is solved on complement of the graph. However, this is not possible if input graphs are restricted so that some specified parameter is bounded. Two interesting examples are tree-width and tree-depth. Note that tree-width of a complement of a graph cannot be bounded by a function of the tree-width of the graph. On the contrary, there are graph parameters in whose this question was already solved, as these graph parameters are (for each constant) closed under taking complements -- neighborhood diversity and modular width, to name just a few. On both these parameters \kPathCover was one of the first considered problems which was showed to be in the \FPT class~\cite{gajarskyLO13}, \cite{lampis12} respectively.

There is a strong connection between the {\sc $L(2,1)$-labeling} problem and {\sc Hamiltonian Anti-Path}. $L(2,1)$-labeling is a labeling of vertices of a graph $G$ so that labels of vertices in distance 1 differ by at least 2, while labels of vertices in distance 2 differ by at least 1. The labels are taken from set $\{0,\ldots,\lambda\}$ and $\lambda$ is called a span. It not hard to see that a graph $G$ with an apex vertex added admits $L(2,1)$-labeling with span $\lambda = n$ if and only if $G$ has Hamiltonian Anti-Path~\cite{bodlaenderKTvL04}.


\paragraph{Our Contribution}
\begin{theorem}\label{thm:kapc-tw}
Let $G$ be a graph of tree-width at most $w$.
The problem of \kAntiPathCover admits an \FPT algorithm parameterized by tree-width.
\end{theorem}


The proof uses the famous closure theorem of Bondy and Chvátal. Most notably we prove the following theorem from which we can derive the previous theorem using e.g.\ known \FPT algorithm of Lampis~\cite{lampis12}.
This exploit an unexpected relation between tree-width and neighborhood diversity. Thus, it naturally rises many questions -- whether similar approach is admissible for other problems besides \kAntiPathCover problem.

\begin{theorem}\label{thm:treewidth}
Let~$G$ be a~graph of tree-width $w$ and further complement~$\cG$ closed under Bondy-Chvátal closure.
It follows that neighborhood diversity of~$\cG$ is bounded by $2^{k} + k$ where $k = 2(w^2 + w)$.
\end{theorem}

Furthermore, we give a natural specializations of the theorems above in Section~\ref{sec:antiHamTD}. The purpose of Section~\ref{sec:antiHamTD} is twofold -- first as tree-depth is more restrictive parameter than tree-width the proof is simpler and second, the analysis of a particular application of Bondy-Chvátal theorem gives more light on the structure of the complement of a graph $G$ that is closed under this closure operator.

\section{Preliminaries}

One of the basic graph operations is taking the complement of a graph. A {\em complement} of a graph $G=(V,E)$ is denoted by $\cG$ and it is the graph $(V, {V \choose 2} \setminus E)$. Throughout the paper we denote by $n$ the number of vertices in the input graph. By the distance between two vertices $u,v$ we mean the length of the shortest path between them in the assumed graph $G$.
We denote the distance by $\dist_G(u,v)$ and omit the subscript if the graph is clear from the context. We extend the notion to sets of vertices in a straightforward manner, that is $\dist_G(U,W) = \min\{\dist_G(u,w) \colon u\in U,\,w\in W\}.$ For further graph related notation we refer reader to the monograph by Diestel~\cite{diestelGT}.

In our approach we repeatedly use the closure theorem of Bondy and Chvátal to increase the number of edges in the complement of a graph (i.e.\ to reduce the number of edges in the given graph). Note that this operation does not increase the tree-width of the input graph.

\begin{theorem}[Bondy-Chvátal closure~\cite{bondyCh76}]
Let  $G = (V,E)$ be a graph of order $|V|\geq3$ and suppose that  $u$ and $v$ are distinct non-adjacent vertices such that $\deg(u)+\deg(v)\ge |V|$.
Now $G$ has a Hamiltonian path if and only if ${(V,E\cup \{u,v\})}$ has a Hamiltonian path.
\end{theorem}

The notion of \emph{tree-width} was introduced by Bertelé and Brioshi~\cite{berteleB73}.
\begin{definition}[Tree decomposition]
A \emph{tree decomposition} of a graph $G$ is a pair $(T,X)$, where ${T=(I,F)}$ is a tree, and $X=\{X_i\mid i\in I\}$ is a family of subsets of $V(G)$ (called bags) such that:
	\begin{itemize}
		\item the union of all $X_i$, $i\in I$ equals $V$,
		\item for all edges $\{v,w\}\in E$, there exists $i\in I$, such that $v,w\in X_i$ and
		\item for all $v\in V$ the set of nodes $\{i\in I\mid v\in X_i\}$ forms a subtree of $T$.
	\end{itemize}
\end{definition}
The \emph{width} of the tree decomposition is $\max(|X_i|-1)$.
The \emph{tree-width} of a graph $\tw{(G)}$ is the minimum width over all possible tree decompositions of the graph~$G$.

\begin{proposition}[\cite{KloksTW}]
Let~$G$ be a~graph with $n$ vertices. There exists an optimal tree decomposition with $O(n)$ bags. Moreover, there is an \FPT algorithm that finds such a decomposition.
\end{proposition}

\begin{definition}[Tree-depth~\cite{nesetrildM12}]
The closure $Clos(F)$ of a forest $F$ is the graph obtained from $F$
by making every vertex adjacent to all of its ancestors. The {\em tree-depth}, denoted as $\td(G)$, of
a graph $G$ is one more than the minimum height of a rooted forest $F$ such that
$G\subseteq Clos(F).$ 
\end{definition}

The last graph parameter needed in this work is the \emph{neighborhood diversity} introduced by Lampis~\cite{lampis12}.
\begin{definition}[Neighborhood diversity]
The \emph{neighborhood diversity} of a graph $G$ is denoted by $\nd{(G)}$ and it is the minimum size of a partition of vertices into classes such that all vertices in the same class have the same neighborhood, i.e.\ ${N(v)\setminus\{v'\}=N(v')\setminus\{v\}}$, whenever
	$v,v'$ are in the same class.
\end{definition}
It can be easily verified that every class of a neighborhood diversity partition is either a clique or an independent set.
Moreover, for every two distinct classes $C,C'$, either every vertex in $C$ is adjacent to every vertex in $C'$,
or there is no edge between $C$ and $C'$. If classes $C$ and $C'$ are connected by edges,
we refer to such classes as \emph{adjacent}. 

It is possible to find the optimal neighborhood diversity decomposition of a given graph in polynomial time~\cite{lampis12}.

\section{Tree-depth}\label{sec:antiHamTD}

We show that using the Bondy-Chvátal closure it is possible on input $G$ and $k$ either decide that $\cG$ is $k$-path coverable or return graph $\overline{H}$ that is equivalent (for the $k$-path coverability) to graph $\cG$ with $\nd(\overline{H})\le 2^{2d} + 2d$.
Furthermore, it is possible to use the \FPT algorithm of Lampis~\cite{lampis12} on the resulting graph to decide whether it is $k$-path coverable or not. This in turn gives an \FPT algorithm for \kAntiPathCover with respect to tree-depth.

\paragraph{Applying Bondy-Chvátal.}
We will apply the Bondy-Chvátal closure from the leaves of a tree-depth decomposition of the input graph $G$ in order to reduce the number of edges in $G$ and either resolve the given question (in the case $G$ becomes an edgeless graph) or impose a structure on the complement of $G$ (after the removal of several edges). Note that every leaf of the decomposition has at most $\td(G)$ neighbors and that the set of leaves spans an edgeless subgraph of $G$ (a clique in \cG). We apply the Bondy-Chvátal closure to a vertex $v$ and all leaves beneath $v$ (denote these as $L$) in the tree-depth decomposition tree. We choose $v$ such that the distance between $L$ and $v$ is 1. We denote $\ell$ the number of nodes, that is $\ell = |L|$.
We denote a {\em height of vertex} $v$ in the tree-depth decomposition as $h(v)$ and define it as follows. Height of a root is set to 0 and for a vertex $v$ let $u$ denote the closest ancestor of $v$ in the tree-depth decomposition we set $h(v) = h(u) + 1$.
Observe that it is possible to add all edges between $L$ and $v$ to~$\cG$ if
$$
n - h(v) - 1 + n - h(v) - \ell \ge n.
$$
That is equivalent to $n > 2h(v) + \ell$.

\begin{lemma}\label{lem:antiHamTDboundND}
Denote $H$ the graph after application of the closure. We claim that $\nd(\overline{H}) \le 2^{2d} + 2d$, where $d = \td(G)$.
\end{lemma}
\begin{proof}
It follows that if the application process stops, then $n\le 2h(v) + \ell$ must hold for all nonleaf vertices. Take $h = \max\{h(v)\colon v \textrm{ nonleaf}\}.$
Note that all (actual) leaves of $H$ form a clique in $\overline{H}$ and thus, $\overline{H}$ is a graph on $n \le 2h+\ell$ vertices with $K_\ell$ as a subgraph. This in turn yields that the distance to clique (the number of vertices to delete such that the resulting graph is a clique) of $\overline{H}$ is at most $2h$. Thus, the neighborhood diversity of $\overline{H}$ is at most $2^{2h} + 2h$. This follows from the fact that there are at most $2^{2d}$ different neighborhoods from the point of view of a clique vertex. This together with trivial fact $h \le \td(G)$ completes the proof.
\end{proof}

Lemma~\ref{lem:antiHamTDboundND} yield the following corollary when known algorithms for finding Hamiltonian path are applied to the resulting graph $\overline{H}$.

\begin{theorem}
The \kAntiPathCover problem admits an \FPT algorithm with respect to parameterization by the tree-depth of the input graph. \qed
\end{theorem}


\section{Tree-width}

In this section we restate and prove Theorem~\ref{thm:treewidth}.

\begin{theorem}[Restate Theorem~\ref{thm:treewidth}]
Let $G$ be a graph of tree-width $w$ and further complement $\cG$ closed under Bondy-Chvátal closure.
It follows that neighborhood diversity of $\cG$ is bounded by $2^{k} + k$ where $k = 2(w^2 + w)$.
\end{theorem}

\begin{proof}
We will prove the graph $\cG$ has a clique $C$ of size at least $n - 2w^2 - w$.
Thus, the graph $\cG$ has at most $2(w^2 + w)$ vertices which are not in $C$.
The bound of $\nd(\cG)$ follows.

Since $\cG$ is closed under Bondy-Chvátal closure, for every edge $\{u,v\} = e \in E(G)$ holds that $\deg_G(u) + \deg_G(v) \geq n$.
Otherwise, it would holds $\deg_{\cG}(u) + \deg_{\cG}(v) \geq n$ and we could add the edge $e$ into $E(\cG)$.
Thus, for every edge $e \in E(G)$ there exists a vertex $v \in e$ such that $\deg_G(v) \geq \frac{n}{2}$.
Let $f: E(G) \rightarrow V(G)$ be a function such that for every $v = f(e)$ holds that $v \in e$ and $\deg_G(v) \geq \frac{n}{2}$.
Let $V_1 = \{f(e) | e \in E(G) \}$.
Note that $V_1$ is a vertex cover of the graph $G$.
Thus, if we remove the set $V_1$ from the graph $\cG$ we obtain a clique.
Moreover, $V_1 \subseteq V_2 = \{v \in V(G) | \deg_G(v) \geq \frac{n}{2}\}$.

It remains to prove that $|V_2| \leq 2(w^2 + w)$.
Let ${\cal T} = (T,X)$ be a tree decomposition of $G$ such that width of ${\cal T}$ is $w$ and $T$ has $n$ nodes.
Let $p$ be a number of all ordered pairs $(v, X_i)$ where $v \in V(G), X_i \in X$ and $v \in X_i$.
We use double counting for $p$.
Since $T$ has at most $n$ nodes and all bags in $X$ contains at most $w + 1$ vertices of $G$, we have
\begin{equation}
\label{eq:UpperBound}
 p \leq n(w + 1).
\end{equation}

Let $v \in V_2$ and $E_v = \{e \in E(G) | v \in e\}$.
Note that $|E_v| = \deg_G(v) \geq \frac{n}{2}$.
Every edge of $G$ has to be in some bag in $X$.
However, there can be only $w$ edges from $E_v$ in one bag in $X$.
Thus, edges from $E_v$ and also the vertex $v$ have to be in at least $\frac{n}{2w}$ bags from $X$.
Therefore, we have lower bound
\begin{equation}
\label{eq:LowerBound}
 |V_2|\frac{n}{2w} \leq p.
\end{equation}

When we join Inequality~\ref{eq:UpperBound} and Inequality~\ref{eq:LowerBound} we get the right upper bound for $V_1$ and $V_2$
\[
 |V_1| \leq |V_2| \leq 2(w^2 + w).
\]
\end{proof}

\section{Conclusions}

We have proven that even through apparently there is no structure in terms of neighborhood diversity on the complements of sparse graphs (having bounded tree-width or tree-depth), the structure after exhaustive application of Bondy-Chvátal closure can be exploited -- the complement has bounded neighborhood diversity.

We would like to ask several vague questions here.
\begin{itemize}
  \item Is it possible to use other graph closure operators to show a connection between tree-width and neighborhood diversity or modular width?
  \item Is it possible to exhibit closer connection between tree-width and modular width trough graph complements?
  \item Does any other non-\MSOt problem besides {\sc $k$-Anti-Path Cover} admit an \FPT algorithm on a graph with  bounded tree-width?
  \item When one assumes parameterization by the tree-width of an input graph it is convenient to approach the problem by the famous theorem of Courcelle~\cite{courcelle90}. Is it possible to extend the theorem for \MSOt for the complementary setting -- i.e.\ to allow quantification over sets of non-edges?
\end{itemize}

\bibliographystyle{eptcs}
\bibliography{src/antiham}

\end{document}